\documentclass[journal]{IEEEtran}
\usepackage{amsmath,amssymb,amsfonts}
\usepackage{algorithmic}
\usepackage{array}
\usepackage[subrefformat=parens,labelformat=parens,caption=false,font=normalsize,labelfont=rm,textfont=rm]{subfig}
\usepackage{textcomp}
\usepackage{stfloats}
\usepackage{url}
\usepackage{verbatim}
\usepackage{graphicx}
\usepackage{cite}
\usepackage{balance}
\usepackage{empheq}
\usepackage{amsthm}

\usepackage{float}
\newtheorem{lemma}{Lemma}
\theoremstyle{remark}
\newtheorem{remark}{Remark}
\theoremstyle{plain}
\newtheorem{theorem}{Theorem}
\usepackage{tikz}
\usepackage[subnum]{cases}

\ifCLASSINFOpdf
\else
\fi
\begin{document}
        %
        \title{A Practical Max-Min Fair Resource Allocation Algorithm for Rate-Splitting Multiple Access}
        \author{Facheng Luo,~\IEEEmembership{Student Member,~IEEE}, Yijie Mao,~\IEEEmembership{Member,~IEEE}
        \thanks{\textit{Corresponding author: Yijie Mao.} 
        \par
        This work has been supported in part by the National Nature Science Foundation of China under Grant 62201347; and in part by Shanghai Sailing Program under Grant 22YF1428400. 
\par 
F. Luo and Y. Mao are with the School of Information Science and Technology, ShanghaiTech University, Shanghai 201210, China (email: \{luofch2022, maoyj\}@shanghaitech.edu.cn).}
	\vspace{-0.5cm}}
    \maketitle
    
    \begin{abstract}
        This letter introduces a novel resource allocation algorithm for achieving max-min fairness (MMF) in a rate-splitting multiple access (RSMA) empowered multi-antenna broadcast channel.
       Specifically, we derive the closed-form solution for the optimal allocation of the common rate among users and the power between the common and private streams for a given practical low-complexity beamforming direction design.
        Numerical results show that the proposed algorithm achieves $90\%$ of the MMF rate on average obtained by the conventional iterative optimization algorithm while only takes an average of $0.1$ millisecond computational time,
        which is three orders of magnitude lower than the conventional  algorithm.
        It is therefore a practical resource allocation algorithm for RSMA.
    \end{abstract}
    

    %
    \IEEEpeerreviewmaketitle

    \vspace{-0.5cm}
    \section{Introduction}
    %
    %
    %
    %
    \IEEEPARstart{R}{ecent} study \cite{outperform} introduced rate-splitting multiple access (RSMA),
    a powerful multiple access scheme that has demonstrated significant potential in future wireless communication systems.
    In the downlink RSMA, the message intended for each user is divided into two components: a common part and a private part.
    The common parts from all users are combined into a single common message and encoded using a shared codebook while each user's private part is encoded independently.
    The encoded symbol streams are linearly precoded and transmitted simultaneously,
    and each receiver decodes both the common stream and its corresponding private stream to recover the intended message\cite{phylayer}.
    In the presence of weak or strong interference, RSMA can adaptively adjust the power and rate allocation, and reduces to either space division multiple access (SDMA) or non-orthogonal multiple access (NOMA).
    As a result, RSMA naturally bridges and outperforms both SDMA and NOMA in terms of spectral efficiency, energy efficiency, user fairness, etc\cite{mao2022rate}.
    
    To enhance the user fairness  (which is known as the max-min fairness--MMF achieved by maximizing the worst-case rate among users) enhancement offered by RSMA, a crucial task is to design resource allocation algorithms for beamforming and common rate allocation, as extensively studied in recent works \cite{mao2022rate}.
    Due to the non-convexity of the MMF resource allocation problems, various suboptimal algorithms have been developed to tackle this challenge, such as the weighted minimum mean-square error (WMMSE) algorithm \cite{robust, yalcin2021}, generalized power iteration (GPI)-based algorithm \cite{kim2022max} and  successive convex approximation (SCA)-based algorithm \cite{mao2018energy, lee2022max, li2023}.
    These algorithms all adhere to the same design paradigm. They aim to tackle the inherent non-convexity of the MMF problem by approximating it as a sequence of convex subproblems, and solving these subproblems in a sequential manner. However, owing to their iterative nature, these algorithms still present practical challenges when applied in real-world scenarios.
    When it comes to designing low-complexity algorithms, there is a lack of research specifically addressing the MMF problem for RSMA.
    \cite{unifying} provides a low-complexity method for precoder design and power allocation for RSMA. But it is specifically tailored for the sum-rate problem and is not applicable for the MMF rate problem.
    To the best of our knowledge, no prior research has managed to derive an analytical solution for the MMF problem for RSMA.

    \vspace{0.1cm}
    In this letter, we bridge this gap by unveiling the analytical solution for solving the MMF resource allocation problem in a two-user RSMA system.
    By fixing the precoding directions using low-complexity zero-forcing (ZF) and multicast beamforming, we calculate the closed-form solutions  for the optimal common rate and power allocation to maximize the worst-case rate among users.
    Our numerical results demonstrate that the proposed algorithm achieves MMF rates that are 90\% of those achieved by the WMMSE/SCA algorithm on average, and outperforms GPI-based algorithm.
    It only takes an average of 0.1 millisecond computational time, which is three orders of magnitude lower than the conventional  algorithms.
    
     \vspace{-2mm}
    
    \section{System Model and Problem Formulation}
   Consider a two-user multiple-input single-output (MISO) broadcast channel (BC) with a base station (BS) simultaneously serving two single-antenna users.
   The BS has $N_t$ transmit antennas and utilizes RSMA to send messages $M_1$ and $M_2$ to user-1 and user-2, respectively.
   The messages are split into $M_{c, 1}$, $M_{p, 1}$ for user-1 and $M_{c, 2}$, $M_{p, 2}$ for user-2. $M_{c, 1}$ and $M_{c, 2}$ are combined into single common message $M_c$.
   The reformed messages $M_{p, 1}$, $M_{p, 2}$ and $M_{c}$ are independently encoded into $s_1$, $s_2$ and $s_c$.
    After linear precoding
    via precoders $\mathbf{p}_1$, $\mathbf{p}_2$, $\mathbf{p}_c$ $\in \mathbb{C}^{N_t\times 1}$, the resulting transmit signal is
    \vspace{-3mm}
    \begin{equation}
        \mathbf{x} = \mathbf{p}_1s_1 + \mathbf{p}_2s_2 + \mathbf{p}_cs_c.
    \end{equation}
    Let $\mathbf{s} = [s_1, s_2, s_c]^T$, $\mathbb{E}\{\mathbf{ss}^H\} = \mathbf{I}$, and $P_i = \|\mathbf{p}_i\|^2$, $i \in \{1, 2, c\}$, we have the power constraint $P_1+P_2+P_c\leq P$.
    
    The received signal at user-$k$ is
    \begin{equation}
        y_k = \mathbf{h}_k^H\mathbf{x} + n_k,\ k\in\{1, 2\},
    \end{equation}
    where $\mathbf{h}_k \in \mathbb{C}^{N_t\times 1}$ is the channel between the BS
    and user-$k$, and it is assumed to be perfectly known at both the BS and user-$k$
    \footnote{In the scenario of imperfect CSIT with bounded error, the MMF problem can be solved by searching in the channel uncertainty region to ensure the worst-case performance.}.
    Without loss of generality, we make the assumption that the channel strength of user-1 is greater than or equal to that of user-2, denoted as $|\mathbf{h}_1| \geq |\mathbf{h}_2|$.
   $n_k$ represents the additive white Gaussian noise (AWGN) encountered by user-$k$, which has a zero mean and unit variance.

    \par
    Initially, each user decodes the common stream $s_c$ with the
    private streams being considered as noise.
    Subsequently, each user utilizes successive interference cancellation (SIC) to eliminate the shared stream $s_c$ from their received signals and then decodes the intended private stream by treating the interference from the other stream as noise.
    The achievable rates to decode the common and private streams at user-$k$ are 
    \begin{align}
        R_{c, k} = \log_2 \left(1 + \frac{|\mathbf h_k^H\mathbf p_c|^2}
        {1 + |\mathbf h_k^H\mathbf p_1|^2 + |\mathbf h_k^H\mathbf p_2|^2}\right), \label{3}\\
        R_k = \log_2 \left(1 + \frac{|\mathbf h_k^H\mathbf p_k|^2}
        {1 + |\mathbf h_k^H\mathbf p_j|^2}\right),\ j \neq k. \label{4}
    \end{align}
    In order to ensure successful decoding of the common stream $s_c$ by both users,
    the achievable common rate must not exceed
    \begin{equation}
        R_c = \min\left\{R_{c, 1}, R_{c, 2}\right\}.
    \end{equation}
    Let $C_k$ be the part of $R_c$ allocated to user-$k$, and it satisfies $0 \leq C_1 + C_2 \leq R_c$.
    The total achievable rate of user-$k$ is
    \begin{equation}
        R_{k, tot} = R_k + C_k.
    \end{equation}
    Hence the MMF rate or symmetric rate is expressed as $R_{\text{MMF}} = \min\{R_{1, tot}, R_{2, tot}\}$.
    
    In this paper, we aim at solving the following MMF optimization problem
    \vspace{-2mm}
    \begin{subequations}\label{maxmin}
        \begin{align}
            &\max_{\mathbf p_1, \mathbf p_2, \mathbf p_c, C_1, C_2} \min \{R_{1, tot}, R_{2, tot}\} \\
            \mathrm{s.t.}\ & C_1 + C_2 \leq R_c, \label{maxminb}\\
            &C_k \geq 0, k \in \{1, 2\}, \label{7c}\\
            &\|\mathbf p_1\|^2 + \|\mathbf p_2\|^2 + \|\mathbf p_c\|^2 \leq P. \label{powerCondition}
        \end{align}
    \end{subequations}
    When $\mathbf{p}_c$ is set to $\mathbf{0}$, RSMA reduces to SDMA.
    When both $\mathbf{p}_1$ and $\mathbf{p}_2$ are set to $\mathbf{0}$, RSMA reduces to Multicast.
    When either $\mathbf{p}_1$ or $\mathbf{p}_2$ is set to $\mathbf{0}$, and (\ref{maxminb}) and (\ref{7c}) are removed, RSMA reduces to NOMA.

    \section{The Proposed Resource Allocation Algorithm}
    \label{sec: algorithm}
    To derive a tractable MMF rate expression, we follow \cite{unifying} to design the precoding directions based on ZF and multicast beamforming.
    \par
    Let $\bar{\mathbf h}_k = \mathbf h_k / \|\mathbf h_k\|$, $k \in \{1, 2\}$.
    The precoding directions for the private streams are designed based on ZF,
    which leads to $|\mathbf h_1^H\mathbf p_2| = 0$, $|\mathbf h_2^H\mathbf p_1| = 0$ and
    $|\mathbf h_k^H\mathbf p_k| = \|\mathbf h_k\|^2 \rho P_k$, where
    $\rho = 1 - |\bar{\mathbf h}_1^H\bar{\mathbf h}_2| ^ 2$.
    The precoder for the common stream is $\mathbf p_c = \sqrt{P_c}\bar{\mathbf p}_c$, and
    $\bar{\mathbf p}_c$ is designed by solving
    \begin{equation}
        \max_{\bar{\mathbf p}_c}\min\{|\bar{\mathbf h}_1^H\bar{\mathbf p}_c|^2,
        |\bar{\mathbf h}_2^H\bar{\mathbf p}_c|^2\},\ \mathrm{s.t.}\ \|\bar{\mathbf p}_c\| = 1.
        \label{pc}
    \end{equation}
    According to \cite{unifying},
    the optimal solution of (\ref{pc}) is achieved when
    $|\bar{\mathbf h}_1^H\bar{\mathbf p}_c| =
    |\bar{\mathbf h}_2^H\bar{\mathbf p}_c|$, which is given by
    \begin{equation}
        \bar{\mathbf p}_c = \frac{1}{\sqrt{2(1 + |\bar{\mathbf h}_1^H\bar{\mathbf h}_2|)}}
        \left(\bar{\mathbf h}_1 + \bar{\mathbf h}_2 e^
        {-j\angle\bar{\mathbf h}_1^H\bar{\mathbf h}_2}\right).
    \end{equation}
    
    For the simplicity of notation, we define $\rho_k = \|{\mathbf h}_k\|^2\rho$ and $\rho_{c, k} = |\mathbf h_k^H\bar{\mathbf p}_c|^2$.
    Based on the above precoding direction design, (\ref{3}) and (\ref{4})
    can be rewritten as
    \begin{equation}\label{10}
        R_{c, k} = \log_2\left( 1 + \frac{\rho_{c, k}P_c}{1 + \rho_k P_k}\right),\ k \in \{1, 2\},
    \end{equation}
    and
    \begin{equation}\label{11}
        R_k = \log_2 (1 + \rho_k P_k),\ k \in \{1, 2\},
    \end{equation}
    respectively.
    
    Substituting (\ref{10}) and (\ref{11}) back to (\ref{maxmin}),
    the optimization variables of (\ref{maxmin}) reduce to the power allocation,
    i.e., $P_1$, $P_2$, $P_c$ and the common rate allocation $C_1$, $C_2$.
    Let $tP$ be the power of private streams and $(1 - t)P$ be the power of common stream, i.e.,
    \begin{equation}
        P_1 + P_2 = tP,\ P_c = (1 - t)P,\ t \in [0, 1].
    \end{equation}
    
   We further design the power allocation $P_1$ and $P_2$
    between the private streams by the simple water-filling (WF) solution \cite{unifying}, which is given by
    \begin{equation}\label{water filling}
        P_k = \max\left\{ \mu -\frac{1}{\rho_k}, 0 \right\},\ k \in \{1, 2\},
    \end{equation}
    where $\mu$ is the water level satisfying $P_1 + P_2 = tP$.
    We will show in the following that the WF solution (\ref{water filling}) leads to a satisfactory MMF rate performance
    and it is helpful for the derivation of the  closed-form optimal solution of $t, C_1, C_2$.
    
    Let $\Gamma = \frac{1}{\rho_2} - \frac{1}{\rho_1}$, according to (\ref{water filling}),
    $P_1$ and $P_2$ are given by
    \begin{equation}\label{WF}
        \left\{\begin{array}{lll}
            P_1 = tP, & P_2 = 0, & \text{if }tP \leq \Gamma, \\
            P_1 = \frac{1}{2}(tP + \Gamma), & P_2 = \frac{1}{2}(tP - \Gamma),
            & \text{if }tP > \Gamma.
        \end{array}\right.
    \end{equation}
    Obviously, RSMA is activated when $0 < P_1 < P$, $0 < P_2 < P$, and $0 < P_c < P$. RSMA reduces to SDMA when $0 < P_1 < P$, $0 < P_2 < P$ and $P_c = 0$.
    RSMA reduces to Multicast when $P_1 = 0$, $P_2 = 0$ and $P_c = P$, and RSMA reduces to NOMA when $0 < P_1 < P$, $P_2 = 0$ and $0 < P_c < P$.
    According to (\ref{WF}), it is easy to obtain that RSMA is activated  when $\frac{\Gamma}{P} < t < 1$. It reduces to NOMA when $0 < t \leq \frac{\Gamma}{P}$,  it reduces to SDMA when $t = 1$ and Multicast when $t = 0$.
    
    With the aforementioned ZF-based precoding direction design and WF power allocation for the private streams, problem (\ref{maxmin}) is simplified to
    \begin{subequations}\label{reducedProblem1}
        \begin{align}
            &\max_{t, C_1, C_2} \min \{R_{1, tot}, R_{2, tot}\} \\
            \mathrm{s.t.}\ & C_1 + C_2 \leq R_c, \\
            &C_k \geq 0, k \in \{1, 2\}, \\
            &0 \leq t \leq 1.
        \end{align}
    \end{subequations}

    Next, we  derive the  solution of (\ref{reducedProblem1}) in the closed form.
    We first introduce the following two Lemmas to find the optimal common rate allocation of RSMA.
    
    \begin{lemma}\label{RcLemma}
        When $tP > \Gamma$, the decoding rates for the common stream at both users are always equal, i.e.,
        \begin{equation}
            R_c = R_{c, 1} = R_{c, 2}.
        \end{equation}
    \end{lemma}
    
    \begin{proof}
        See Appendix \ref{RcProof}.
    \end{proof}
    
    \begin{lemma}\label{maxminLemma}
        When RSMA is activated, the objective function of (\ref{maxmin}) is equivalent to
        \begin{equation}\label{17}
            R_{\text{MMF}}^{\text{RSMA}} = 
            \frac{1}{2}\left(R_1 + R_2 + R_c -
            \max\left\{R_1 - R_2 -R_c, 0\right\}\right).
        \end{equation}
        Furthermore, the optimal $C_1$, $C_2$, and MMF rate are given by
        \begin{align}
            &\left\{
            \begin{array}{ll}
                C_1^* = 0, \\
                C_2^* = R_c, \\
                R_{\text{MMF}}^{\text{RSMA}} = R_2 + R_c,
            \end{array}\right.
            &{\text{if}}\ R_c \leq R_1 - R_2, \\
            &\left\{
            \begin{array}{ll}
                C_1^* = (R_c - R_1 + R_2) / 2, \\
                C_2^* = (R_c + R_1 - R_2) / 2, \\
                R_{\text{MMF}}^{\text{RSMA}} = (R_1 + R_2 + R_c) / 2,
            \end{array}\right.
            &{\text{if}}\ R_c > R_1 - R_2.
        \end{align}
    \end{lemma}
    
    \begin{proof}
        See Appendix \ref{maxminProof}.
    \end{proof}
    
    \begin{remark}
     Despite the fact that we fix the precoding directions and adopt the WF solution, the conclusion of Lemma \ref{maxminLemma} is applicable to any specific precoding direction and power allocation within the RSMA activated region.
    \end{remark}
    
    Based on Lemma \ref{maxminLemma}, we could obtain that the optimal $C_1$ and $C_2$ can be represented as a function of $R_1$, $R_2$, and $R_c$.
    We could therefore replace the objective function of (\ref{reducedProblem1}) by (\ref{17}) when RSMA is activated.
    Let $R_1 = R_2 = 0$ in (\ref{17}), we could obtain the MMF rate of Multicast as $R_{\text{MMF}}^{\text{Multicast}} = R_c / 2$.
    Let $R_c = 0$ in (\ref{17}), the MMF rate of SDMA is $R_{\text{MMF}}^{\text{SDMA}} = R_2$.
    Note that Lemma \ref{maxminLemma} is not suitable for NOMA.
    The MMF rate of NOMA is $R_{\text{MMF}}^{\text{NOMA}} = \min\{R_1, R_c\}$.
    Then the MMF rates of different strategies can be formulated as
    \begin{equation}\label{20}
        R_{\text{MMF}}^{\text{X}} = 
        \left\{\begin{array}{ll}
            \begin{array}{l}
                \frac{1}{2}(R_1 + R_2 + R_c - \\
                \max\left\{R_1 - R_2 -R_c, 0\right\}),
            \end{array} & \text{if X} = \text{RSMA} \\
            \min\{R_1, R_c\}, & \text{if X} = \text{NOMA} \\
            R_2, & \text{if X} = \text{SDMA} \\
            R_c / 2, & \text{if X} = \text{Multicast} \\
        \end{array}\right.
    \end{equation}
    
    Based on (\ref{20}), problem (\ref{reducedProblem1}) is equivalently transformed to
    \vspace{-5mm}
    \begin{subequations}\label{reducedProblem2}
        \begin{align}
            & \max_{t, \text{X}} R_{\text{MMF}}^{\text{X}} \\
            &\text{s. t. } 0 \leq t \leq 1,
        \end{align}
    \end{subequations}
    where X should be taken over all four strategies.
    
    \begin{theorem}\label{theorem1}
        The optimal $t$ of problem (\ref{reducedProblem2}) falls within the following six points:
        \begin{numcases}{t^* \in }
            \frac{2\rho_{c, 2}P - \rho_1\Gamma - \rho_2\Gamma}{\rho_1 P - \rho_2 P + 2\rho_{c, 2}P}, \label{RSMA1} \\
            \frac{\frac{1}{2}\rho_2\Gamma - \rho_{c, 2}P - 1}{\rho_2P - 2\rho_{c, 2}P} - 
            \frac{1}{\rho_1P} - \frac{\Gamma}{2P}, \label{RSMA2} \\
            \Gamma / P, \label{NOMA1} \\
            \rho_{c, 2} / (\rho_1 + \rho_{c, 2}), \label{NOMA2} \\
            1, \label{SDMA} \\
            0. \label{Multicast}
        \end{numcases}
        In particular, if $t^*$ is equal to (\ref{RSMA1}) or (\ref{RSMA2}),
         RSMA is activated.
        If $t^*$ is equal to (\ref{NOMA1}) or (\ref{NOMA2}),
        RSMA reduces to NOMA.
        If $t^*$ is equal to (\ref{SDMA}),
        RSMA reduces to SDMA.
        If $t^*$ is equal to (\ref{Multicast}),
        RSMA reduces to Multicast.
    \end{theorem}
    
    \begin{proof}
        See Appendix \ref{proof1}.
    \end{proof}
    
    According to Theorem \ref{theorem1}, we then obtain the optimal solution by substituting the six candidate optimal $t$ back to (\ref{reducedProblem2}) and selecting the one that achieves the highest MMF rate.
    Our algorithm always chooses the best strategy from RSMA, SDMA, NOMA and Multicast.
    As the optimal $t^*$, $C_1^*$ and $C_2^*$ are all calculated in closed form, the computational complexity of our algorithm is $O(N_t)$. 
    
    At high SNR, since $\Gamma / P \rightarrow 0$ as $P \rightarrow +\infty$, NOMA is not a suitable strategy anymore.
    Let $t_1 = \frac{2\rho_{c, 2}P - \rho_1\Gamma - \rho_2\Gamma}{\rho_1 P - \rho_2 P + 2\rho_{c, 2}P}$
    and $t_2 = \frac{\frac{1}{2}\rho_2\Gamma - \rho_{c, 2}P - 1}{\rho_2P - 2\rho_{c, 2}P} - 
    \frac{1}{\rho_1P} - \frac{\Gamma}{2P}$ be the optimal $t$ candidates for RSMA. If $0 < |\bar{\mathbf h}_1^H\bar{\mathbf h}_2| < 1$ and $\|\mathbf{h}_1\| \neq \|\mathbf{h}_2\|$,  the following inequalities hold
    \begin{equation}
        \begin{aligned}
            &\lim_{P \rightarrow +\infty}
            \left(R_{\text{MMF}}^{\text{RSMA}}|_{t = t_1} - R_{\text{MMF}}^{\text{SDMA}}|_{t = 1}\right) \\
            = &\log_2\left(\frac{2\rho_1\rho_{c, 2}}{\rho_2(\rho_1 - \rho_2 + 2\rho_{c, 2})}\right) > 0
        \end{aligned}
    \end{equation}
    and
    \begin{align}
        \begin{aligned}
            &\lim_{P \rightarrow +\infty}
            \left(R_{\text{MMF}}^{\text{RSMA}}|_{t = t_2} - R_{\text{MMF}}^{\text{SDMA}}|_{t = 1}\right) \\
            =& \frac{1}{2}\log_2\left(\frac{\rho_1\rho_{c, 2}^2}{\rho_2^2(2\rho_{c, 2} - \rho_2)}\right) > 0,
        \end{aligned}
    \end{align}
    which implies that RSMA brings a constant MMF rate gain over SDMA in the high SNR regime.
    
    It can also be easily verified that the MMF rate gain of RSMA over Multicast grows unbounded as $P \rightarrow +\infty$.
    
    Therefore, the optimal $t$ of Problem (\ref{reducedProblem2}) will be either $\frac{2\rho_{c, 2}}{\rho_1 - \rho_2 + 2\rho_{c, 2}}$ or $\frac{\rho_{c, 2}}{2\rho_{c, 2} - \rho_2}$, when $P \rightarrow +\infty$, and only RSMA is activated. 
    
    \vspace{-3mm}
    \section{Numerical Results}

    In this section, we evaluate the performance of the proposed algorithm and compare it with some other methods listed below:
    \begin{enumerate}
        \item WMMSE-based algorithm \cite{robust}. Its worst-case computational complexity is $\mathcal{O}((\log(\epsilon^{-1})(KN_t)^{3.5})$.
        \item SCA-based algorithm \cite{mao2018energy}. Its worst-case computational complexity is $\mathcal{O}(\log(\epsilon^{-1})(KN_t)^{3.5})$.
        \item GPI-based algorithm \cite{kim2022max}. Its worst-case computational complexity is $\mathcal{O}((\#\gamma)\log(\epsilon^{-1})KN_t^{3})$,
        where $\#\gamma$ is the number of exhaustive search of Lagrangian multiplier $\gamma$ as specified in \cite{kim2022max}.
    \end{enumerate}
    Two-user cases with perfect CSIT are considered in all the experiments.
    
   \par
   In Fig. \ref{figure11}, the BS has $N_t=2$ transmit antennas.
   Fig. \ref{figure11}(a) and Fig. \ref{figure11}(b) illustrate
    the MMF rate versus SNR of different strategies.
    The channel for each user has independent identically distributed complex Gaussian entries, i.e.,
    $\mathbf{h}_k\sim \mathcal{CN}(0, \sigma_k^2)$. The results are averaged over 100 random channels.
    We observe that the proposed algorithm attains an average of $92.9\%$ and $93.0\%$ of the MMF rates obtained by RSMA in Fig. \ref{figure11}(a) and Fig. \ref{figure11}(b), respectively.
    The rates obtained by the proposed algorithm outperform those of other multiple access schemes (except RSMA) based on SCA/WMMSE at moderate and high SNR regimes.
    Moreover, our algorithm outperforms GPI-based algorithm in terms of the MMF rate.

    \par 
    Fig. \ref{figure11}(c) shows the CPU time versus SNR, averaged over $100$ random channels.
    The CPU time of the proposed  algorithm is three orders of magnitude lower than that of WMMSE/SCA/GPI algorithm, since
    there is no iterative step in the proposed algorithm.

    \begin{figure*}[htbp]
        \centering
        \includegraphics[width=\textwidth]{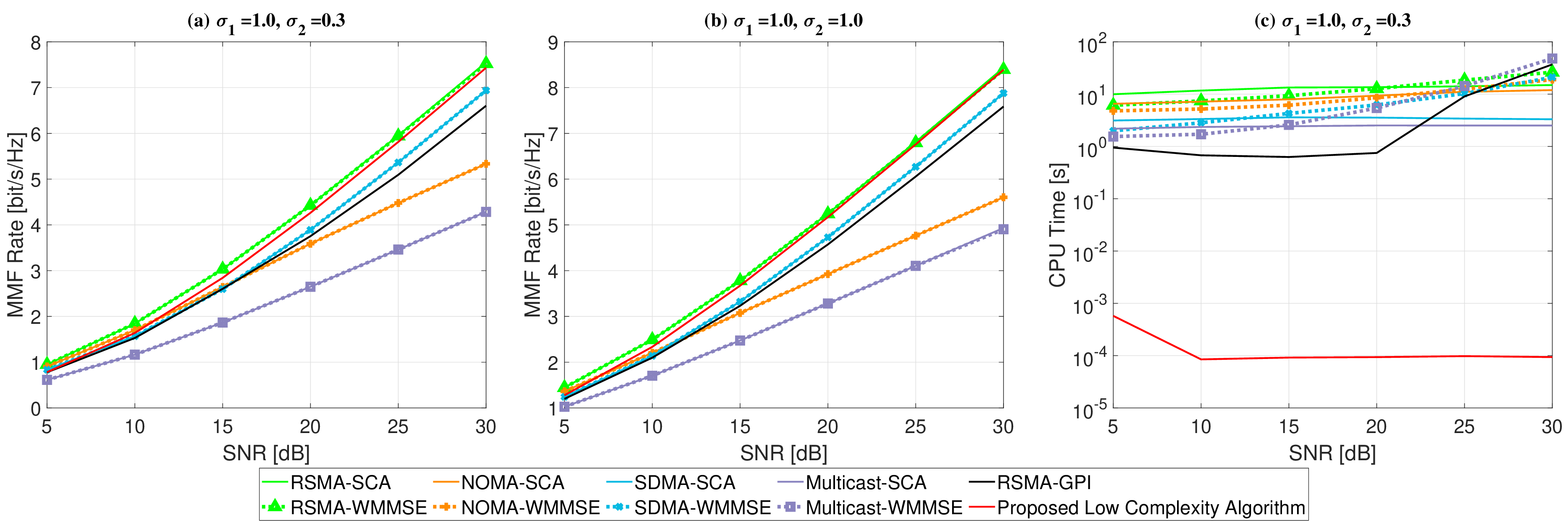} 
        \caption{MMF rate or CPU time versus SNR, averaged over 100 random channel realizations.}
        \label{figure11}
        \vspace{-0.5cm}
    \end{figure*}

    Fig. \ref{optimalT} shows the optimal $t$ that maximizes the
    MMF rate obtained by our proposed algorithm.
    Following \cite{unifying}, we use specific channels, i.e., $\mathbf{h}_1 = \frac{1}{\sqrt{2}}[1, 1]^H$ and
    $\mathbf{h}_2 = \frac{\gamma}{\sqrt{2}}[1, e^{j\theta}]^H$.
    The $x$-axis and $y$-axis of Fig. \ref{optimalT} are defined as $\rho = 1 - |\bar{\mathbf h}_1^H\bar{\mathbf h}_2| ^ 2$
    and $\gamma_{\text{dB}} = 20\log_{10}(\gamma)$, respectively.
    As SNR increases, the actived region of NOMA becomes smaller while the region of RSMA becomes larger.
    At high SNR, for example SNR = 30dB, almost all the region when $-15\text{dB} < \gamma_{\text{dB}} < 0$ is RSMA.
    This implies that RSMA is preferred in the high SNR regime. This finding aligns with the theoretical analysis presented in Section \ref{sec: algorithm}.

    \begin{figure*}[htbp]
        \centering
        \subfloat{\includegraphics[width=0.32\textwidth]{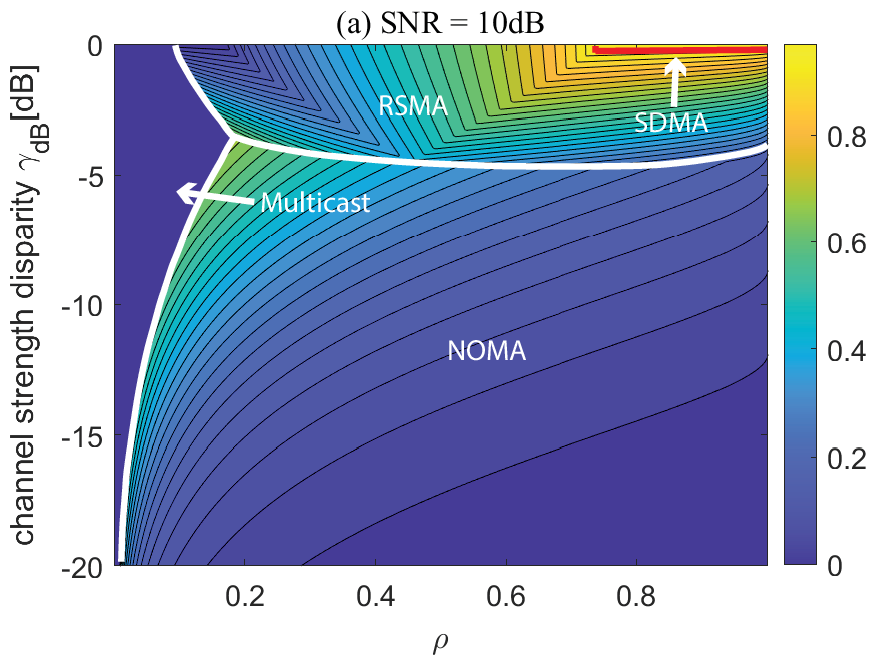}}
        \hfill
        \subfloat{\includegraphics[width=0.32\textwidth]{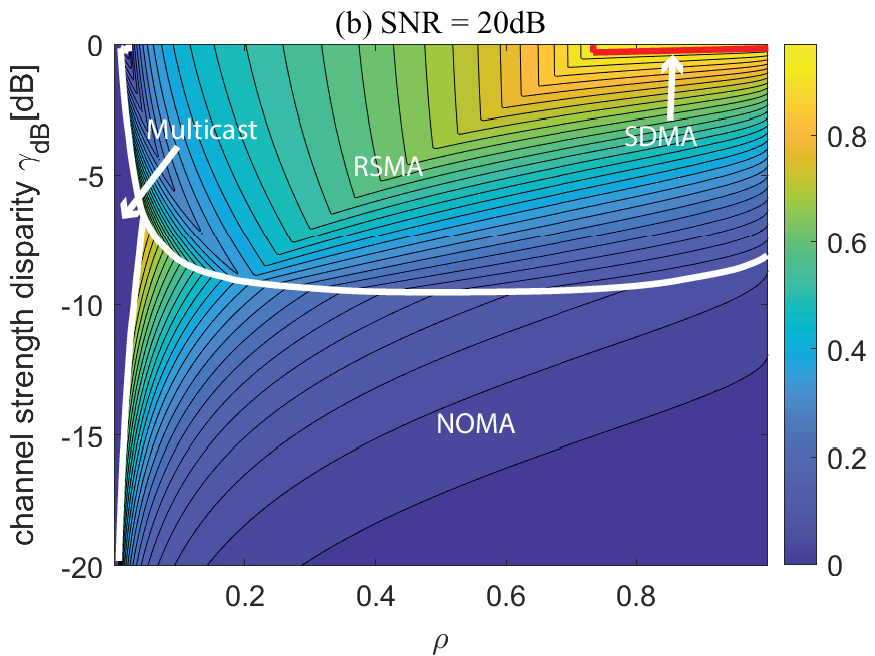}}
        \hfill
        \subfloat{\includegraphics[width=0.32\textwidth]{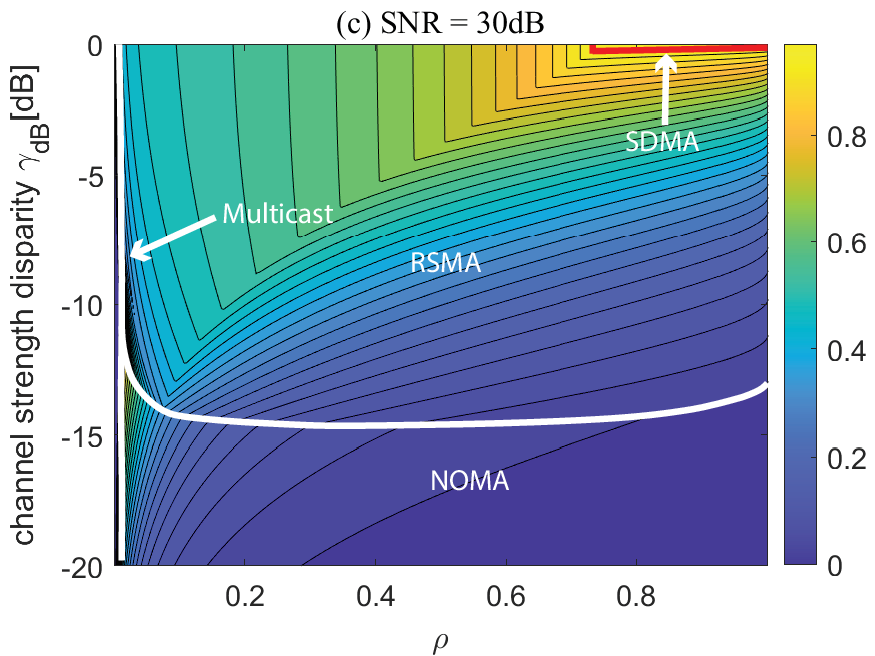}}

        \vspace{-4mm}
        \caption{Optimal $t$ obtained by our proposed algorithm.}
        \label{optimalT}
        \vspace{-5mm}
    \end{figure*}

    Fig. \ref{relativeGain} illustrates the relative rate gain of RSMA compared to the dynamic switching approach between SDMA, NOMA, and Multicast.
    The relative gain is calculated by 
    \begin{equation*}
        \frac{R_{\text{MMF}}^{\text{RSMA}} - \max\{R_{\text{MMF}}^{\text{SDMA}}, R_{\text{MMF}}^{\text{NOMA}}, R_{\text{MMF}}^{\text{Multicast}}\}}{\max\{R_{\text{MMF}}^{\text{SDMA}}, R_{\text{MMF}}^{\text{NOMA}}, R_{\text{MMF}}^{\text{Multicast}}\}}\times 100\%.
    \end{equation*}
    The percentages in parentheses represent MMF rate gains of RSMA over SDMA, NOMA, and Multicast, respectively.
    RSMA brings explicit gains over other strategies in the activated region of RSMA.
    
    \begin{figure}[htbp]\label{relativeGain30dB}
        \centering
        \includegraphics[width=0.35\textwidth]{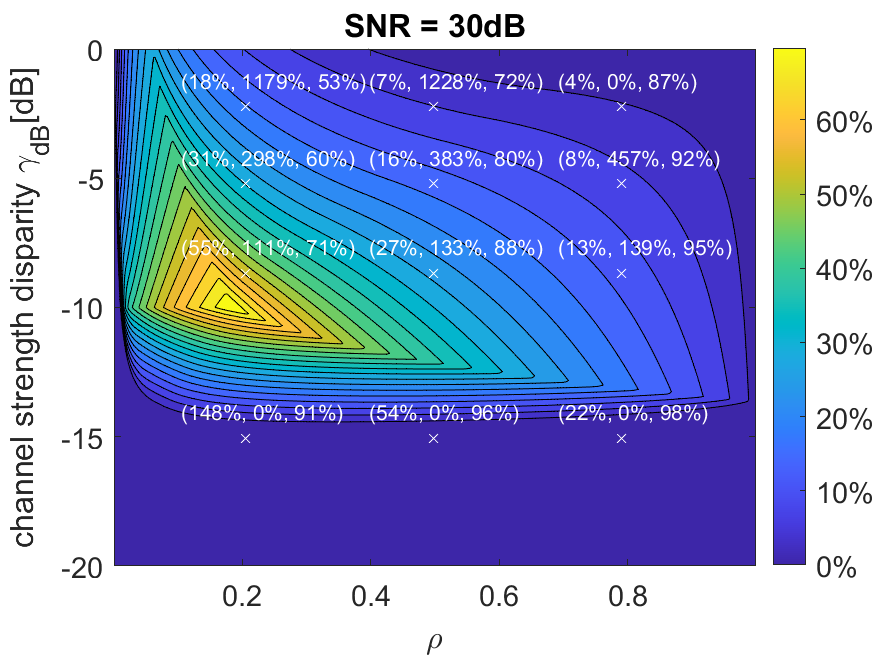}
        \vspace{-4mm}
        \caption{Relative MMF rate gain of RSMA compared to dynamic switching between SDMA, NOMA and Multicast. The percentages in parentheses represent MMF rate gains over SDMA, NOMA and Multicast, respectively.}
        \label{relativeGain}
        \vspace{-5mm}
    \end{figure}

    \vspace{-3mm}
    \section{Conclusion}
    \vspace{-2mm}
    In this work, we propose a novel power allocation and common rate allocation algorithm to achieve the MMF of RSMA.  With a practical precoding direction design based on zero-forcing and water-filling power allocation for the private streams,
    we obtain the closed-form solution for the optimal allocation of common rate and power between the common and private streams.
    This is the first work that derives a closed-form solution to MMF problem of RSMA.
    The computational complexity of the proposed algorithm is three orders of magnitude lower than that of the WMMSE/SCA/GPI algorithm.
    The MMF rate of our algorithm is within $90\%$ of WMMSE/SCA and outperforms GPI-based algorithm.
    Therefore, we draw the conclusion that the proposed resource allocation algorithm for RSMA is practical and efficient.
    The future work of this letter is to expand the analysis to the $K$-user setting.

    \vspace{-3mm}
    \bibliographystyle{IEEEtran}{}
    \bibliography{IEEEexample}

    
    %
     \vspace{-5mm}
    
    \appendices
    \section{Proof of Lemma \ref{RcLemma}}\label{RcProof}
    \vspace{-1mm}
    \begin{proof}
        According to (\ref{10}), we only need to prove that $\frac{\rho_{c, 1}}{1 + \rho_1 P_1} = \frac{\rho_{c, 2}}{1 + \rho_2 P_2}$
        when $\Gamma/P < t < 1$.
        
        If $tP > \Gamma$, according to (\ref{WF}), we have
        \begin{equation}\label{25}
            P_1 = \frac{1}{2}(tP + \Gamma),\ P_2 = \frac{1}{2}(tP - \Gamma).
        \end{equation}
        Hence,
        \begin{equation}
            \frac{\rho_{c, 1}}{1 + \rho_1 P_1} =
            \frac{2|\bar{{\mathbf h}}_1^H\bar{\mathbf p}_c|^2}
            {\rho tP + \frac{1}{\|{\mathbf h}_1\|^2} + \frac{1}{\|{\mathbf h}_2\|^2}},
        \end{equation}
        \begin{equation}
            \frac{\rho_{c, 2}}{1 + \rho_2 P_2} =
            \frac{2|\bar{{\mathbf h}}_2^H\bar{\mathbf p}_c|^2}
            {\rho tP + \frac{1}{\|{\mathbf h}_1\|^2} + \frac{1}{\|{\mathbf h}_2\|^2}}.
        \end{equation}
        Since $\bar{\mathbf p}_c$ is obtained when $|\bar{\mathbf h}_1^H\bar{\mathbf p}_c| = |\bar{\mathbf h}_2^H\bar{\mathbf p}_c|$,
        it can be easily derived that
        \begin{equation}
            \frac{\rho_{c, 1}}{1 + \rho_1 P_1} =
            \frac{\rho_{c, 2}}{1 + \rho_2 P_2}.
        \end{equation}
    \end{proof}

     \vspace{-3mm}
    
    \section{Proof of Lemma \ref{maxminLemma}}\label{maxminProof}
    \begin{proof}
        According to $\min\{a, b\} = \frac{1}{2}(a + b - |a - b|)$,
        the minimum rate of the two users can be rewritten as
        \begin{equation}\label{absoluteValue}
            \begin{aligned}
                &\min \{R_{1, tot}, R_{2, tot}\} \\
                =& \frac{1}{2}(R_1 + C_1 + R_2 + C_2 - |R_1 + C_1 - R_2 - C_2|) \\
                =& \frac{1}{2}(R_1 + R_2 + R_c - | R_1 - R_2 + (C_1 - C_2) |).
            \end{aligned}
        \end{equation}
        
        Since $C_1 + C_2 \leq R_c$, $C_1 \geq 0$, $C_2 \geq 0$, $C_1 - C_2$
        can be bounded by
        \begin{equation}\label{bound}
            -R_c \leq C_1 - C_2 \leq R_c,
        \end{equation}
        where $C_1 - C_2 = R_c$ when $C_1 = R_c$, $C_2 = 0$ and $C_1 - C_2 = -R_c$ when $C_1 = 0$, $C_2 = R_c$.
        
        By combining (\ref{absoluteValue}) and (\ref{bound}), we have
        \begin{equation}
            \begin{aligned}
                &\max\min \{R_{1, tot}, R_{2, tot}\} \\
                =& \frac{1}{2}\left(R_1 + R_2 + R_c -
                \min_{C_1, C_2}\left| R_1 - R_2 + (C_1 - C_2) \right|\right) \\
                =& \frac{1}{2}\left(R_1 + R_2 + R_c -
                \max\left\{R_1 - R_2 -R_c, 0\right\}\right) \\
                = &\begin{cases}
                    R_2 + R_c, &{\text{if}}\ R_c \leq R_1 - R_2, \\
                    \frac{1}{2}(R_1 + R_2 + R_c), &{\text{if}}\ R_c > R_1 - R_2.
                \end{cases}
            \end{aligned}
        \end{equation}
        The corresponding optimal $C_1$ and $C_2$ are
        \begin{align}
            &\left\{
            \begin{array}{ll}
                C_1^* = 0, \\
                C_2^* = R_c,
            \end{array}\right.
            &{\text{if}}\ R_c \leq R_1 - R_2, \\
            &\left\{
            \begin{array}{ll}
                C_1^* = \frac{1}{2}(R_c - R_1 + R_2), \\
                C_2^* = \frac{1}{2}(R_c + R_1 - R_2),
            \end{array}\right.
            &{\text{if}}\ R_c > R_1 - R_2.
        \end{align}
    \end{proof}

    \vspace{-6mm}
    
    \section{Proof of Theorem \ref{theorem1}}\label{proof1}
    \begin{proof}
        As $t$ ranges from $0$ to $1$, the MMF rate is a piecewise
        function about $t$. Therefore the optimal $t$ must be one of
        critical points. Obviously $0$, $1$ and $\Gamma / P$ are the
        critical points.
        
        Next we give the critical points for RSMA, where $\Gamma / P < t < 1$.
        In this case, we use (\ref{25}) for power allocation.
        According to Lemma \ref{RcLemma} and Lemma \ref{maxminLemma},
        when $R_c > R_1 - R_2$, i.e., $t < \frac{2\rho_{c, 2}P - \rho_1\Gamma - \rho_2\Gamma}
        {\rho_1 P - \rho_2 P + 2\rho_{c, 2}P}$,
        \vspace{-2mm}
        \begin{equation}\label{RSMArate1}
            \begin{aligned}
                R_{\text{MMF}}^{\text{RSMA}}
                &=\frac{1}{2}\left(R_1 + R_2 + R_c\right) \\
                &=\frac{1}{2}\log_2 ((1 + \rho_1 P_1)(1 + \rho_2 P_2 + \rho_{c, 2}P_c)).
            \end{aligned}
        \end{equation}
        Then we can easily find the critical point of equation (\ref{RSMArate1})
        \begin{equation}
            \begin{aligned}
                &\arg\left\{\frac{\partial R_{\text{MMF}}^{\text{RSMA}}}{\partial t} = 0\right\} \\
                = &\frac{\frac{1}{2}\rho_2\Gamma - \rho_{c, 2}P - 1}
                {\rho_2P - 2\rho_{c, 2}P} - 
                \frac{1}{\rho_1P} - \frac{\Gamma}{2P}.
            \end{aligned}
        \end{equation}
        When $R_c \leq R_1 - R_2$,
        \begin{equation}\label{RSMArate2}
            R_{\text{MMF}}^{\text{RSMA}}
            = R_2 + R_c 
            = \log_2 (1 + \rho_2 P_2 + \rho_{c, 2}P_c).
        \end{equation}
        The critical point of equation (\ref{RSMArate2}) is $\frac{2\rho_{c, 2}P - \rho_1\Gamma - \rho_2\Gamma}
        {\rho_1 P - \rho_2 P + 2\rho_{c, 2}P}$.
        
        When NOMA is considered, i.e., $t < \Gamma / P$, since $R_{c, 1} \geq R_{c, 2}$ when $t \leq \min\{\Gamma / P, 1\}$, we have
        \begin{equation}
            R_c = \min\{R_{c, 1}, R_{c, 2}\} = R_{c, 2} = \log_2(1 + \rho_{c, 2}P_c).
        \end{equation}
        Therefore, the MMF rate of NOMA is
        \begin{equation}
            \begin{aligned}
                R_{\text{MMF}}^{\text{NOMA}} &= \min\{R_1, R_c\} \\
                &= \min\left\{\log_2(1 + \rho_1P_1), \log_2(1 + \rho_{c2}P_c)\right\}.
            \end{aligned}
        \end{equation}
        The critical point can be found at
        $\rho_{c, 2} / (\rho_1 + \rho_{c, 2})$ where $\log_2(1 + \rho_1P_1) = \log_2(1 + \rho_{c2}P_c)$.
    \end{proof}
    

    \ifCLASSOPTIONcaptionsoff
    \newpage
    \fi

\end{document}